\documentclass[aps,twocolumn,showpacs,amsmath,amssymb]{revtex4}

\usepackage{amsthm}
\usepackage{latexsym}
\usepackage{amsfonts}
\usepackage{bbm,dsfont}

\newtheorem{theorem}{Theorem}

\newtheorem{proposition}{Proposition}

\newcommand{\R}{\mathbb R}
\newcommand{\Z}{\mathbb Z}
\newcommand{\C}{\mathbb C}

\newcommand{\hi}{\mathcal{H}} 
\newcommand{\mi}{\mathcal{M}} 

\newcommand{\tr}[1]{\mathrm{tr}\left[#1\right]} 
\newcommand{\kb}[2]{|#1\,\rangle\langle\,#2|} 

\newcommand{\Qo}{\mathsf{Q}} 
\newcommand{\Po}{\mathsf{P}} 
\renewcommand{\P}{\mathcal P} 
\newcommand{\A}{\mathsf{A}}

\def\<{\langle}
\def\>{\rangle}
\def\d{{\mathrm d}}
\newcommand{\fii}{\varphi}

\newcommand{\CHI}[1]{\ensuremath{ \chi\raisebox{-1ex}{$\scriptstyle #1$} }}
\newcommand{\ov}{\overline}
\newcommand{\I}{{\cal I}}
\newcommand{\lin}{{\rm lin}}

\begin{document}

\title{Complete Characterization of Pure Quantum Measurements and Quantum Channels}

\author{Juha-Pekka Pellonp\"a\"a}
\email{juhpello@utu.fi}
\affiliation{Turku Centre for Quantum Physics, Department of Physics and Astronomy, University of Turku, FI-20014 Turku, Finland}

\begin{abstract}
We give a complete characterization for pure quantum measurements, i.e., for POVMs which are extremals in the convex set of all POVMs. Such measurements are free from classical noise. The characterization is valid both in discrete and continuous cases, and also in the case of an infinite Hilbert space. We show that sharp measurements are clean, i.e.\ they cannot be irreversibly connected to another POVMs via quantum channels and thus they are free from any additional quantum noise.
We exhibit an example which demonstrates that this result could also be approximately true for pure measurements.
\end{abstract}


\maketitle


In the modern formalism of quantum mechanics, measurements
are described by positive operator valued measures (POVMs) 
which have found ample applications in various areas of quantum physics, 
ranging from quantum theory of open systems to detection, estimation and quantum information theories.
POVMs generalize the traditional concept of an 'observable', a selfadjoint operator or a projection valued measure (PVM), which turned out to be a too restrictive idealization to efficiently descibe actual experimental settings such as fuzzy position and momentum measurements or photon counting and phase measurements in quantum optics \cite{PSAQT82, kirjat}.

A fundamental problem is to characterize
the most precise and informative measurements of a physical quantity.
One crucial property of such optimal measurements is the lack of noise; classical or quantum.
Therefore, the present letter focuses on the determination of noise-free measurements.
Here we consider two types of noise:
classical noise associated with the randomness due to fluctuations in the measuring procedure and quantum noise due to the possibility of irreversibly manipulating the state before a measurement (using a quantum channel).

Similarly to quantum states, POVMs form a convex set if the measurement outcome space and the Hilbert space are fixed. A convex combination $a\Po(X)+(1-a)\Po'(X)$, $0<a<1$, corresponds to a classical randomization or mixing between two (or more) POVMs $\Po$ and $\Po'$. Such mixing is a source of classical noise.
 Extremal or {\it pure} POVMs do not admit any convex decompositions and are free from classical noise
\cite{PSAQT82}.
For a finite dimensional system, a simple criterion for extremality can be given
 \cite{DA} and Chiribella {\it et al}.\ \cite{ChDASc} showed that all pure measurements are concentrated on a finite number of outcomes. However, in the infinite case, there exist pure (nonsharp) POVMs with continuous measurement outcome spaces \cite{holevo, HePe}. 

In this letter, we fully characterize all pure measurements using a diagonalization technique of Hyt\"onen {\it et al}.\ \cite{HyPeYl}. This result is a generalization of the finite dimensional characterization \cite{DA}.
We also introduce a simple polynomial method for finding pure POVMs in continuous cases. This method is very useful in many areas of quantum physics, e.g.\ in continuous variable quantum information.

Finally, we show that any PVM $\Po$ can be connected to any ($\Po$-continuous) POVM $\Po'$ via a quantum channel $\Phi$ (i.e.\ $\Phi^*(\Po(X))=\Po'(X)$), or in other words, $\Po'$ is a pre-processing of $\Po$ \cite{Bu}.
The pre-processing can change the POVM irreversibly, reducing the information from the measurement. Our result shows that PVMs are {\it clean} \cite{ch} or 'undisturbed' in the sense that they are not irreversibly connected to another POVMs. 
Thus they do not have any additional 'extrinsical'  quantum noise \cite{Bu}.

Let us briefly recall the mathematical description of quantum measurements via \emph{(normalized) positive operator valued measures} (POVMs) \cite{PSAQT82}. Consider a quantum system with a Hilbert space $\hi$ and suppose that the measurement outcomes form a set $\Omega$. 
A POVM is a function $\Po$ which associates to each (m.) subset $X\subseteq\Omega$ a positive operator $\Po(X)$ acting on $\hi$ \cite{sigma}. It is required that for every state (a density matrix) $\varrho$, the mapping
$X\mapsto \tr{\varrho\Po(X)}$
is a probability distribution. Especially, $\Po$ satisfies the normalization condition $\Po(\Omega)=I$. 
The number $\tr{\varrho\Po(X)}$ is the probability of getting a measurement outcome $x$ belonging to $X$, when the system is in the state $\varrho$ and the measurement $\Po$ is performed. 
A POVM $\Po$ is a \emph{projection valued measure} (PVM), or a \emph{sharp} POVM, if $\Po(X)^2=\Po(X)$ for all $X\subseteq\Omega$.  
It is easy to see that PVMs are pure (see a new proof below).

Fix an orthonormal (ON) basis $\{ e_n \}_{n=1}^{\dim\hi}$ of $\hi$, let $V$ be the vector subspace of $\hi$ consisting of finite linear combinations of the basis vectors $e_n$, and let $V^\times$ be the vector space of formal series $\sum_{n} c_n e_n$, $c_n\in\C$.
Define a probability distribution $\mu(X):=\sum_{n=1}^{\dim\hi}\lambda_n\<e_n|\Po(X)e_n\>$ where $\lambda_n>0$ and
$\sum_n\lambda_n=1$ \cite{mu}. Any POVM $\Po$ can be diagonalized in the following way:

\begin{theorem} \label{th1}
a)
There are (m.) mappings $n(x)\in\{0,1,\ldots,\dim\hi\}$ and $d_k(x)\in V^\times\setminus\{0\}$ such that
$$
\<\fii|\Po(X)\psi\>=\int_X \sum_{k=1}^{n(x)} \<\fii|d_k(x)\>\<d_k(x)|\psi\>\d\mu(x),\hspace{0.2cm}\fii,\,\psi\in V.
$$
\\
b) There are (m.) maps $\psi_m(x)\in\hi_{n(x)}\subseteq\hi$ such that
\begin{eqnarray*}
\Po(X)&=&\sum_{n,m=1}^{\dim\hi}\int_X\<\psi_n(x)|\psi_m(x)\>\d\mu(x)\kb{e_n}{e_m} \\
&=&\Big(\sum_m\kb{\CHI X\psi_m}{e_m}\Big)^*\Big(\sum_m\kb{\CHI X\psi_m}{e_m}\Big)
\end{eqnarray*}
(weakly) and the set of linear combinations of vectors $\CHI X \psi_m$ is dense in 
$\int^\oplus_\Omega\hi_{n(x)}\d\mu(x)$ (the minimal Kolmogorov decomposition).
\\
c) $\Po(X)=J^*\ov\Po(X)J$ where $J:=\sum_{m=1}^{\dim\hi}\kb{\psi_m}{e_m}$ (is an isometry, i.e., $J^*J=I$) and $\ov\Po(X)=\CHI X$ is the (canonical) PVM on  
$\hi_{\ov\Po}:=\int^\oplus_\Omega\hi_{n(x)}\d\mu(x)$ (the minimal Naimark dilation).
\\
d) $\Po$ is a PVM if and only if $\{\psi_m\}_{m=1}^{\dim\hi}$ is an ON basis of $\hi_{\ov\Po}$. Then $\hi_{\ov\Po}$ can be identified with $\hi$ (i.e.\ $JJ^*=I$ and $J$ is a unitary operator). \cite{Dirac}
\end{theorem}

Here $\CHI X$ is the characteristic function of $X$ and
 $\int^\oplus_\Omega\hi_{n(x)}\d\mu(x)$ is the direct integral of Hilbert spaces $\hi_{n(x)}$ where $\hi_l$ is an $l$-dimensional Hilbert space spanned by vectors $e_1,\,e_2,\ldots,e_l$, $\hi_0:=\{0\}$ 
and $\hi_\infty:=\hi$ (if $\dim\hi=\infty$) \cite{dir}.

\begin{proof}
a) and b) follow from Theorems 4.5 and 5.1 of \cite{HyPeYl} by defining $\psi_m(x):=\sum_{k=1}^{n(x)}\<d_k(x)|e_m\>e_k$, and c) follows from b) and  Theorem 3.6 of \cite{HyPeYl}.
Finally, d) follows from Corollary 5.2 of \cite{HyPeYl}.
\end{proof}

By defining operators
$
\A(x):=\sum_{k=1}^{n(x)}\kb{e_k}{d_k(x)}=\sum_{m=1}^{\dim\hi}\kb{\psi_m(x)}{e_m}
$
one can write \cite{sesq}
\begin{equation}\label{eq1}
\Po(X)=\int_X \A(x)^*\A(x)\d\mu(x).
\end{equation}
\begin{theorem}\label{th2}
A POVM $\Po$ is pure if and only if, for any (bounded decomposable \cite{dec}) operator $D=\int_\Omega^\oplus D(x)\d\mu(x)$ on $\hi_{\ov\Po}$
the condition $\int_\Omega \<\psi_n(x)|D(x)\psi_m(x)\>\d\mu(x)=0$ for all $n,m$
implies that $D=0$.
\end{theorem}

The above condition can also be written in the form
$
J^*DJ=\int_\Omega \A(x)^*D(x)\A(x)\d\mu(x)=0,
$
or in the form
\begin{equation}\label{eq2}
\int_\Omega\sum_{k,l=1}^{n(x)}\<e_k|D(x)e_l\>\kb{d_k(x)}{d_l(x)}\d\mu(x)=0.
\end{equation}
Hence, Theorem \ref{th2} is a ('continuous' and infinite) generalization of \cite{DA}.
Formally, one could say that $\Po$ is pure iff 'the (overcomplete) system of generalized coherent states $d_k(x)$ generates a linearly independent set of operators $\kb{d_k(x)}{d_l(x)}$  in the sense that \eqref{eq2} implies $D=0$.'

Since $\<\psi_n|\psi_m\>=\int_\Omega\<\psi_n(x)|\psi_m(x)\>\d\mu(x)=\delta_{nm}$
one can define a projection $\P:=JJ^*=\sum_{m=1}^{\dim\hi}\kb{\psi_m}{\psi_m}$ and the above condition equals $\P D\P=0$. If $\Po$ is a PVM then $\P=I$ and $\P D\P=D$ so that any PVM is pure.

\begin{proof}
Suppose that there exists a nonzero bounded $\int_\Omega^\oplus D(x)\d\mu(x)$ such that
$\int_\Omega \<\psi_n(x)|D(x)\psi_m(x)\>\d\mu(x)=0$ for all $n,\, m$.
Redefining $D$ as $i(D-D^*)$ (if $D^*\ne D$)
and then scaling $D$ by $1/\|D\|$, one may assume that $D^*=D$, $\|D\|\le 1$ and, thus, $D_\pm:=I\pm D\ge 0$ and $D_+\ne D_-$.
Since vectors $\CHI X \psi_m$ span $\hi_{\ov\Po}$ there exists a (m.) set $X'$ and $n',\,m'$ such that
$\int_{X'}\<\psi_{n'}(x)|D_+(x)\psi_{m'}(x)\>\d\mu(x)\ne\int_{X'}\<\psi_{n'}(x)|D_-(x)\psi_{m'}(x)\>\d\mu(x)$ implying that POVMs $\Po_\pm(X):=\sum_{n,m}\int_X\<\psi_n(x)|D_\pm(x)\psi_m(x)\>\d\mu(x)\kb{e_n}{e_m}$ are distinct and $\Po=(\Po_++\Po_-)/2$ so that $\Po$ is not pure.

Suppose then that $\Po$ is not pure, that is, of the form $\Po=(\Po_++\Po_-)/2$ for some POVMs $\Po_\pm$, $\Po_+\ne \Po_-$. Now $\Po_\pm(X)\le 2\Po(X)$
so that (by Theorem \ref{th1} and \cite{mu})
$
\Po_{\pm}(X)=\sum_{n,m=1}^{\dim\hi}\int_X\<\psi^\pm_n(x)|\psi^\pm_m(x)\>\d\mu(x)\kb{e_n}{e_m} = \int_X \A(x)_{\pm}^*\A(x)_{\pm}\d\mu(x)
$
where
$\A(x)_\pm:=\sum_{m}\kb{\psi^\pm_m(x)}{e_m}$.
In addition, $\<\fii|\Po_\pm(X)\fii\>\le 2\<\fii|\Po(X)\fii\>$ (for all $\fii\in V$) implies that $\|\A(x)_\pm\fii\|\le\sqrt2\|\A(x)\fii\|$ (for all $\fii\in V$) holds for $\mu$-a.a.\ $x\in\Omega$.
Hence, one can define bounded (well-defined) operators $C_\pm(x)$ on $\hi_{n(x)}$ as follows: (a) define
$C_\pm(x)\big(\A(x)\fii\big):=\A(x)_\pm\fii$, (b) extend $C_\pm(x)$ to the closure of $\A(x)V$, and 
(c) extend $C_\pm(x)$ to the whole fiber $\hi_{n(x)}$ by setting $C_\pm(x)$ to zero on the orthogonal complement of the closure of $\A(x)V$.
Define then (linear) operators $C_\pm$ by $(C_\pm\psi)(x):=C_\pm(x)\psi(x)$ where 
$\psi$ is a linear combination of vectors $\CHI X\psi_m$.
Since $\|C_\pm(x)\|\le\sqrt2$, $C_\pm(\CHI X\psi_m)=\CHI X\psi^\pm_m$ and vectors
$\CHI X\psi_m$ span $\hi_{\ov\Po}$, one can extend $C_\pm$ to the whole space $\hi_{\ov\Po}$ and $C_\pm=\int_\Omega^\oplus C_\pm(x)\d\mu(x)$ is bounded.
Define then $D_\pm(x):=C_\pm(x)^*C_\pm(x)$ to get
$
\Po_{\pm}(X)=\sum_{n,m=1}^{\dim\hi}\int_X\<\psi_n(x)|D_\pm(x)\psi_m(x)\>\d\mu(x)\kb{e_n}{e_m}
$
since $C_\pm(x)\psi_m(x)=\psi^\pm_m(x)$.
From the assumption $\Po_+\ne \Po_-$ one gets
$
D:=\int_\Omega^\oplus[D(x)_+-D(x)_-]\d\mu(x)\ne 0.
$
But, for all $n,m$,
$
\int_\Omega \<\psi_n(x)|D(x)\psi_m(x)\>\d\mu(x)=\<e_n|[\Po_+(\Omega)-\Po_-(\Omega)]e_m\>=\delta_{nm}-\delta_{nm}=0.
$
\end{proof}

From Theorem \ref{th2} one gets the following necessary conditions for $\Po$ to be pure. Let $\Po$ be a pure POVM:
\begin{itemize}
\item For any bounded (m.) function $f(x)\in\C$, the condition
$\int_\Omega f(x)\d\Po(x)=0$ implies $f(x)=0$ (for $\mu$-a.a.\ $x\in\Omega$).
(Put $D=f$ in Theorem \ref{th2}.)
\item For any {\it fixed} $X\subseteq\Omega$,  the condition
$\int_X \<\psi_n(x)|D(x)\psi_m(x)\>\d\mu(x)=0$ for all $n,m$
implies that $D(x)=0$ for $\mu$-a.a.\ $x\in X$.
(Replace $D$ with $\CHI X D$ in Theorem \ref{th2}.)
\item For any disjoint (m.) sets $X_1,\ldots, X_p$ such that $\Po(X_i)\ne 0$
the condition $\sum_{i=1}^p c_i \Po(X_i)=0$ implies $c_1=\cdots=c_p=0$, i.e., effects $\Po(X_i)$ are linearly independent. (Now $D=\sum_i c_i\CHI{X_i}$.)
 \end{itemize}
Next we introduce a simple concrete polynomial method for finding pure (continuous) quantum measurements.


Assume that $\hi_{\ov\Po}=L^2(\Omega,\mu,\hi_l)$ \cite{p}.
Usually in physically relevant 'continuous' cases $\Omega\subseteq\R^p$ and, by choosing suitable coordinates, $\Omega$ is of the form $\I:=\I_1\times\cdots\times\I_p$ where $\I_i\subseteq\R$ is an interval. (Without restricting generality, we may even assume that any $\I_i$ is either $[-1,1]$, $[0,\infty)$ or $\R$.) 
Moreover, in practice $\d\mu(x)=w_1(x^1)\cdots w_p(x^p)\d x^1\cdots\d x^p$ (where $x=(x^1,\ldots,x^p)$ and any 'weight function' $w_i(x^i)>0$ is integrable) and an ON basis of $L^2(\I,\mu,\hi)$ is $\{f^1_{k_1}\otimes\cdots\otimes f^p_{k_p}\otimes e_n\}$ ($n=1,\ldots,l$) where $\{f^i_{k_i}\}$ is an ON polynomial basis of $L^2(\I_i,w_i(x^i)\d x^i,\C)$ \cite{poly}. 
For simplicity, we assume that $n(x)\equiv 1$, i.e..\ $l=1$ and $\hi_{\ov\Po}\cong L^2(\I,\mu,\C)$. Hence $\psi_m(x)\in\C$.
{\it Suppose then that any $\psi_m(x)$ is a polynomial.} (Otherwise  $\psi_m(x)$ can be approximated by a polynomial.)
From Theorem \ref{th2} we get the {\it polynomial method}:
\begin{proposition}\label{pr1}
$\Po$ is pure if the linear span of $\big\{\overline{\psi_n(x)}\psi_m(x)\big\}_{n,m}$ contains all polynomials.
\end{proposition}

\begin{proof}
$\Po$ is pure if and only if, for any (m.) bounded $\lambda$, 
$\int_\Omega \lambda(x)\overline{\psi_n(x)}\psi_m(x)\d\mu(x)=0$ for all $n,m$ implies $\lambda=0$.
Assume that $\lin\{\overline{\psi_n(x)}\psi_m(x)\}_{n,m}$ contains all polynomials. Since (any bounded) $\lambda\in L^2(\I,\mu,\C)$ and the space of all polynomials is dense in $L^2(\I,\mu,\C)$ we get the claim of the proposition.
\end{proof}

For example, let $\Qo(X)=\CHI X$, $X\subseteq\R$, be the PVM of the canonical position operator $(Q\psi)(x)=x\psi(x)$ of a particle moving on a line. Using Hermite functions $h_n(x)=c_n H_n(x)e^{-x^2/2}$ we can write $\Qo(X)=\sum_{m,n=0}^\infty\int_X h_n(x)h_m(x)\d x\kb{h_n}{h_m}$.
Let $P_k:=I-\kb{h_k}{h_k}$ be a projection, $\hi=P_kL^2(\R,\d x,\C)$, $\d\mu(x)=e^{-x^2}\d x$, and $\Qo_k(X)=P_k\Qo(X)P_k$ a POVM with vectors $\psi_n(x)=c_n H_{n}(x)$, $n\ne k$.
If, say, $k=2$ then $\Qo_2$ is pure by the polynomial method (since $\{H_n(x)H_m(x)\}_{n\ne 2,m\ne 2}$ contains at least one polynomial of each degree: $H_0(x)H_m(x)$ is a polynomial of degree $m\ne 2$ and $H_1(x)H_1(x)$ is a polynomial of degree 2). 
Similarly, using the polynomial method, one easily sees that the measurement of the (quantum optical) $Q$-function \cite{holevo} and
 the canonical phase measurement \cite{HePe} are pure. 
Next we show that any PVM $\Po$ can be connected to any ($\Po$-continuous) POVM $\Po'$ via a channel $\Phi$ \cite{ch}.


Let $\Po$ and $\Po'$ be POVMs with the same outcome space $\Omega$ but acting possibly different (separable) Hilbert spaces $\hi$ and $\hi'$. Let $\{e'_n\}$ be an ON basis of $\hi'$.
Similarly, as is the case of $\Po$ (see Theorem 1), we let $\mu'(X)$, $n'(x)$, $\psi_n'(x)$, etc.\ denote the corresponding maps related to $\Po'$.
Suppose that there exists a channel $\Phi$ such that $\Phi^*(\Po(X))=\Po'(X)$ for all $X$. Then, if $\Po(X)=0$ one has $\Po'(X)=0$ so that $\d\mu'(x)=w(x)\d\mu(x)$
where $w(x)\ge 0$; we say that  $\Po'$ is {\it $\Po$-continuous.}
If $\Po'$ is $\Po$-continuous, one can absorb $\sqrt{w(x)}$ into functions $\psi'_n(x)$ and redefine $\psi'_n(x)$ to be $\sqrt{w(x)}\psi'_n(x)$. Hence, without restricting generality, we may assume that $\mu'(X)=\mu(X)$. 

\begin{theorem}\label{th3}
a) \
There exists a channel $\Phi$ such that $\Phi^*(\Po(X))=\Po'(X)$ for all $X$ if and only if $\Po'$ is $\Po$-continuous, there exist vectors $v^s_n$ in a (separable) Hilbert space $\mi$ 
such that 
 $\sum_{s=1}^{\dim\hi}\<v_n^s|v_m^s\>=\delta_{nm}$, and there exists an isometry $W=\int_\Omega W(x)\d\mu(x)$ from $\hi'_{\ov\Po'}$ to $\mi\otimes\hi_{\ov\Po}$ such that
$W(x)\psi'_n(x)=\sum_{s=1}^{\dim\hi}v_n^s\otimes\psi_s(x)$. \cite{St} \\
b) \
If $\Po$ is sharp and $\Po'$ \emph{any} $\Po$-continuous POVM then there exists a channel such that $\Phi^*(\Po(X))=\Po'(X)$ for all $X$. \\
c) \
If there exists a channel $\Phi$ such that $\Phi^*(\Po(X))=\Po'(X)$ for all $X$ then $\Phi^*(B)=\ov\Phi^*(JBJ^*)$ for all bounded operators $B$ on $\hi$, where $\ov\Phi$ is a channel connecting the dilation $\ov \Po$ to $\Po'$. 
\end{theorem}

\begin{proof}
a) \
Any channel $\Phi$ has a Kraus decomposition $\Phi^*(B)=\sum_{k=1}^{N}A_k^*BA_k$ (ultraweakly) where $B:\,\hi\to\hi$ and $A_k:\,\hi'\to\hi$ are bounded operators and $N\le\dim\hi\dim\hi'$ \cite{Kraus}.
Let $\mi$ be any Hilbert space with an ON basis $\{f_k\}_{k=1}^\infty$.
Defining $v_n^s:=\sum_{k}\<e_s|A_ke'_n\>f_k\in\mi$ one gets
\begin{eqnarray}\label{hep}\nonumber
\<e_n'|\Phi^*(B)e_m'\>
&=&\sum_{k=1}^N\sum_{s,t=1}^{\dim\hi}\<e_n'|A_k^*e_s\>\<e_s|Be_t\>\<e_t|A_ke_m'\>\\
&=&\sum_{s,t=1}^{\dim\hi}\<v_n^s|v_m^t\>\<e_s|Be_t\>
\end{eqnarray}
and $\sum_{s=1}^{\dim\hi}\<v_n^s|v_m^s\>=\delta_{nm}$. Especially,
$\|v_n^s\|^2\le\sum_{s}\|v_n^s\|^2=1$. 
Conversely, if there exist a Hilbert space $\mi$ (with an ON basis $\{f_k\}_{k=1}^{\dim\mi}$) and vectors $v_n^s\in\mi$ such that $\sum_{s=1}^{\dim\hi}\<v_n^s|v_m^s\>=\delta_{nm}$, one can define a map $\Phi^*$ by equation \eqref{hep} and (bounded) operators $A_k:=\sum_{s,n}\<f_k|v_n^s\>\kb{e_s}{e'_n}$, $\sum_k A_k^*A_k=I$, to get $\Phi^*(B)=\sum_{k=1}^{\dim\mi}A_k^*BA_k$ so that $\Phi$ is a channel.
From \eqref{hep} 
we get the following fact:
There exists a channel $\Phi$ such that $\Phi^*\Po=\Po'$ iff there exist vectors $v^s_n\in\mi$, $\sum_{s=1}^{\dim\hi}\<v_n^s|v_m^s\>=\delta_{nm}$, such that (for all $X$)
\begin{eqnarray}\label{sep}\nonumber
&&\int_X\<\psi_n'(x)|\psi_m'(x)\>\d\mu(x) =\sum_{s,t=1}^{\dim\hi}\<v_n^s|v_m^t\>\<\psi_s|\CHI X\psi_t\>\\
&&=\sum_{s,t=1}^{\dim\hi}\int_X\<v_n^s\otimes\psi_s(x)|v_m^t\otimes\psi_t(x)\>\d\mu(x).
\end{eqnarray}
If there exists a (decomposable) isometry $W$ such that 
$W(x)\psi'_n(x)=\sum_{s=1}^{\dim\hi}v_n^s\otimes\psi_s(x)$
then \eqref{sep} clearly follows. Conversely, if \eqref{sep} holds then
$
\<\CHI X\psi_n'|\CHI{\tilde X}\psi_m'\>=\sum_{s,t=1}^{\dim\hi}\<v_n^s\otimes\CHI X\psi_s|v_m^t\otimes\CHI{\tilde X}\psi_t\>
$
and, since vectors $\CHI X\psi_n'$ span $\hi_{\ov\Po'}$ by Theorem 1,
there exists an isometry $W:\,\hi'_{\ov\Po'}\to\mi\otimes\hi_{\ov\Po}$ such that
$W(\CHI X\psi'_n)=\sum_s v_n^s\otimes(\CHI X\psi_s)$.
But now $W$ commutes with a PVM $\CHI X$ so that it commutes with the von Neumann algebra $L^\infty(\Omega,\mu,\C)$ and is thus decomposable \cite{von}.

b) Let $\Po$ be a PVM and $\Po'$ a $\Po$-continuous POVM.
Let $\Omega'$ be a (m.) subset of $\Omega$ such that $\Omega'$ consists of points $x$ for which $n(x)\ne 0$. (Note that if $n(x)=0$ then $n'(x)=0$ (for $\mu$-a.a.\ $x$) by the $\Po$-continuity of $\Po'$.) Let $\mu|_{\Omega'}$ be a restriction of $\mu$ to $\Omega'$, and let
$\{\eta_s\}_{s=1}^M$ be an ON basis of $L^2(\Omega',\mu|_{\Omega'},\C)$ (which is separable since $\hi_\Po=\hi_{\ov\Po}$ is separable by Theorem 1).
Extend $\{\eta_s e_1\}_{s=1}^M$ to an ON basis $\{\psi_s\}_{s=1}^{\dim\hi}$
of $\hi_\Po$ (note that this forces $M\le\dim\hi$).
Since functions $x\mapsto\<e_k|\psi'_n(x)\>$ belong to $L^2(\Omega',\mu|_{\Omega'},\C)$
they can be represented as $L^2$-convergent series with respect to the basis $\{\eta_s\}_{s=1}^M$ and, hence, $\<e_k|\psi'_n(x)\>=\sum_{s=1}^{\dim\hi}c_{kn}^s\<e_1|\psi_s(x)\>$ where $\sum_{s,k}\overline{c_{kn}^s}{c_{km}^s}=\delta_{nm}$.
Define vectors $v^s_n:=\sum_{k=1}^\infty c_{kn}^s f_k\in\mi$. 
Now
$
\<\psi_n'|\CHI X\psi_m'\>
=\sum_{s,t}\<v_n^s|v_m^t\>\<\psi_s\CHI X|\psi_t\>
$
so that there exists a channel $\Phi$ such that $\Phi^*\Po=\Po'$ by \eqref{sep}.

c) Assume that $\Phi^*\Po=\Po'$ and let $v^s_n$'s be the vectors associated to $\Phi$. 
It is easy to check that $\ov\Phi^*(\ov B):=\sum_{n,m=1}^{\dim\hi'}\sum_{s,t=1}^{\dim\hi}\< v_n^s|v_m^t\>\<\psi_s|\ov B\psi_t\>\kb{e_n'}{e_m'}$ (where $\ov B$ is a bounded operator on $\hi_{\ov P}$) is a channel for which $\ov\Phi^*(\ov\Po(X))=\Po'(X)$ 
and $\Phi^*(B)=\ov\Phi^*(JBJ^*)$.
\end{proof}

It is reasonable to expect that b) of Theorem \ref{th3} is valid approximately for pure POVMs $\Po$ as the following example demonstrates.

Consider the canonical phase POVM \cite{HePe}
$$\Po(X)=\int_X\kb\theta\theta\frac{\d\theta}{2\pi}=\sum_{n,m=0}^\infty\int_X e^{i(n-m)\theta}\frac{\d\theta}{2\pi}\kb n m$$ where $|\theta\>:=\sum_n e^{in\theta}|n\>$ is the Susskind-Glogower phase state and $X\subseteq\Omega=[0,2\pi)$.
Now $x=\theta$, $\d\mu(\theta)=\d\theta/(2\pi)$, $e_n=|n-1\>$, $n(\theta)\equiv1$, $d_1(\theta)=|\theta\>$, and $\psi_n(\theta)=e^{-in\theta}|0\>\cong e^{-in\theta}$.
Let $\Po'$ be any $\Po$-continuous POVM, i.e.\ $\d\mu'(\theta)=\d\theta$ and
$
\Po'(X)=\sum_{n,m=1}^{\dim\hi'}\int_X\<\psi'_n(\theta)|\psi'_m(\theta)\>\d\theta\kb{e'_n}{e'_m}.
$
Since $\psi_n'\in L^2(\Omega,\d\theta,\hi')$ it has the Fourier series
$\psi_n'(\theta)=\sum_{s=-\infty}^\infty\tilde v_n^s e^{-is\theta}$.
Let $N<\dim\hi'+1$ and $\epsilon>0$. One can pick an $M>0$ such that
$$\left|\int_X\<\psi'_n(\theta)|\psi'_m(\theta)\>\d\theta-\int_X\<\psi^M_n(\theta)|\psi^M_m(\theta)\>\d\theta\right|<\epsilon$$ for all $X$ and $n,m\le N$ where $\psi_n^M(\theta):=\sum_{s=-M}^M \tilde v_n^s e^{-is\theta}$. Define a POVM 
$$
\Po'_M(X)=\sum_{n,m=1}^{\dim\hi'}\int_X\<\psi^M_n(\theta)|\psi^M_m(\theta)\>\d\theta\kb{e'_n}{e'_m}
$$
which is not necessarily normalized. However, one can consider $\Po'_M$ as an approximation of $\Po'$. Since $e^{-iM\theta}\psi_n^M(\theta)=\sum_{s'=-M}^M \tilde v_n^{s'} e^{-i(s'+M)\theta}=\sum_{s=0}^{2 M}v_n^s\psi_s(\theta)$, $v_n^s:=\tilde v_n^{s-M}$, we get from a) of Theorem \ref{th3} that there exists a (possibly nonunital) channel $\Phi_M$ such that $\Phi_M^*(\Po(X))=\Po'_M(X)\approx\Po'(X)$.

In conclusion, we have shown that the traditional observables, PVMs, have a special role among quantum measurements, namely, they are clean and free from any additional extrinsical quantum noise. However, the above example suggests that this result could also be approximately true for all pure measurements and, hence, the most accurate quantum measurements should be described by pure POVMs. 
The physically significant POVMs usually satisfy certain properties of covariance with respect to a symmetry group of the theory \cite{PSAQT82}.
For example, quantum optical phase measurements are described by POVMs which transform covariantly with respect to phase shifts generated by the number operator. However, covariant phase POVMs are never sharp.
Theorem \ref{th2} and Proposition \ref{pr1} provide powerful tools (i) for constructing pure POVMs descibing 'canonical' (covariant) measurements of physical quantities (phase, time, angle, etc.) 
and (ii) for studying whether (or not) a POVM associated to an actual measurement scheme (e.g.\ homodyne detection in quantum optics) is pure. 


\begin{thebibliography}{99}
\bibitem{ch}
A channel $\varrho\mapsto\Phi(\varrho)$ is a completely positive (trace-preserving linear) map between state spaces associated to physical systems. Its dual map $\Phi^*$ (i.e.\ $\tr{\varrho\Phi^*(B)}=\tr{\Phi(\varrho)B}$) is identity-preserving (unital): $\Phi^*(I)=I$. Physically, the condition $\Phi^*\Po=\Po'$ (i.e.\ $\tr{\varrho\Po'(X))}=\tr{\Phi(\varrho)\Po(X)}$) between POVMs means that to get the measurement outcome statistics of $\Po'$ in the state $\varrho$ one can equally measure $\Po$ in the state $\Phi(\varrho)$. 
A POVM $\Po$ is (pre-processing) clean if, for any ($\Po$-continuous) POVM $\Po'$ and a channel $\tilde\Phi$ such that $\Po=\tilde\Phi^*\Po'$ there exists a channel $\Phi$ such that $\Po'=\Phi^*\Po$. Hence,
a clean POVM cannot be obtained by (irreversibly) manipulating the state before the measurement and then measuring some other POVM.
\bibitem{mu}Indeed, $\mu$ can be any positive measure such that $\Po$ is absolutely continuous with respect to $\mu$ \cite{HyPeYl}, Remark 3.8.
In the discrete case $\mu$ is a (weighted) counting measure and all integrals reduce to sums.
Moreover, $V$ can be extended to a Hilbert space $H_1$ and then we have a triplet of Hilbert spaces $H_1\subseteq\hi\subseteq H_{-1}$ \cite{HyPeYl}, Section 6.
\bibitem{Dirac}
If $\Omega=\R$ and $\Po$ is the PVM associated to an operator $P=P^*$ then $Pd_k(x)=xd_k(x)$ so that $x$ is an 'eigenvalue' of $P$ with 'multiplicity' $n(x)$ \cite{HyPeYl}.
Note that d) can be {\it formally} expressed in the Dirac's formalism as follows:
'the resolution of the identity' $\int_\Omega \sum_{k=1}^{n(x)} |d_k(x)\>\<d_k(x)|\d\mu(x)=I$ defines a PVM iff $\<d_k(x)|d_l(y)\>=\delta_{kl}\delta_x(y)$
where $\delta_{kl}$ is the Kronecker's delta and $\delta_x$ is the 'Dirac's delta' concentrated on $x$. 
Then {\it formally} 
$\A(x)\A(y)^*=\delta_x(y)\sum_{k=1}^{n(x)}\kb{e_k}{e_k}$ and 
 \begin{eqnarray*}
&& \<\fii|\psi\>=\iint\<\fii(x)|\A(x)\A(y)^*\psi(y)\>\d\mu(y)\d\mu(x) \\
&&=\sum_m\int\<\fii(x)|\psi_m(x)\>\d\mu(x)\int\<\psi_m(y)|\psi(y)\>\d\mu(y)
\end{eqnarray*}
which implies that
$\{\psi_m\}$ is an ON basis.
Moreover, 
$\Po(X)^2=\int_X\int_X\A(x)^*\A(x)\A(y)^*\A(y)\d\mu(x)\d\mu(y)=\int_X\left[\int_X\delta_x(y)\A(x)^*\A(y)\d\mu(y)\right]\d\mu(x)=\Po(X)$.
\bibitem{dir}
Recall that $\int^\oplus_\Omega\hi_{n(x)}\d\mu(x)$ is a (closed) subspace of $L^2(\Omega,\mu,\hi)\cong L^2(\Omega,\mu,\C)\otimes\hi$, the space of ($\mu$-square integrable) 'wave functions' $\psi:\,\Omega\to\hi$, such that $\psi(x)\in\hi_{n(x)}$ (for $\mu$-a.a.\ $x$). Usually, when $\Omega$ is a manifold and $\mu$ is a 'volume' form, 
$\int^\oplus_\Omega\hi_{n(x)}\d\mu(x)$ can be seen as a space of (square integrable) vector fields (e.g.\ spinors), that is, sections of a vector bundle over $\Omega$ with fibers $\hi(x)$. The map $J$ can be viewed as a transformation from the 'matrix mechanics' to the 'wave mechanics' generated by the POVM $\Po$, and $\hi$ can be identified with the '$\Po$-reprensentation' space $\hi_{\Po}:={\rm Im}\,J=\P\hi_{\ov\Po}=\big(\sum_m\kb{\psi_m}{\psi_m}\big)L^2(\Omega,\mu,\hi)$. 
\bibitem{sesq}Note that $\A(x)$ is defined on $V$ (for $\mu$-a.a.\ $x$). The formula \eqref{eq1} must be understood as a sesquilinear form $V\times V\to\C$, that is, 'sandwich' eq.\ \eqref{eq1} between vectors   
$\<\fii|\cdots|\psi\>$ 
as in a) of Theorem \ref{th1}.
\bibitem{dec}Recall that a decomposable operator $D=\int_\Omega^\oplus D(x)\d\mu(x)$ operates as $(D\psi)(x)=D(x)\psi(x)$ and it is bounded if $\|D(x)\|\le c<\infty$ (for $\mu$-a.a.\ $x$).
\bibitem{p}
Collecting the fibers of the same dimension together,
$\int^\oplus_\Omega\hi_{n(x)}\d\mu(x)=L^2(\Omega_1,\mu_1,\hi_1)\oplus
L^2(\Omega_2,\mu_2,\hi_2)\oplus\cdots\oplus
L^2(\Omega_N,\mu_N,\hi_N)$,
$N=\dim\hi$,
 where $\Omega_s:=n^{-1}(\{s\})$ and $\mu_s$ is the restriction of $\mu$ to $\Omega_s$ (i.e.\ $\mu_s(X):=\mu(X)$ for (m.) $X\subseteq\Omega_s$). If $\Po$ is pure then $\int_{\Omega_s} \<\psi_n(x)|D_s(x)\psi_m(x)\>\d\mu(x)=0$ for all $n,m$
implies that $\int^\oplus_{\Omega_s}D_s(x)\d\mu(x)=0$ so that it is not very restrictive to consider the case $\hi_{\ov\Po}=L^2(\Omega,\mu,\hi_l)$.
\bibitem{poly}
Usually any $f_{k_i}^i(y)$ is a classical ON polynomial: when $\I_i=[-1,1]$ (and $w_i(y)=(1-y)^\nu(1+y)^\mu$), $f_{k_i}^i(y)$ is a Jacobi polynomial $P_n^{(\nu,\mu)}(y)$, when $\I_i=[0,\infty)$ (and $w_i(y)=y^\nu e^{-y}$), $f_{k_i}^i(y)$ is an associated Laguerre polynomial $L^\nu_n(y)$, and when $\I_i=\R$ (and $w_i(y)=e^{-y^2}$), $f_{k_i}^i(y)$ is a Hermite polynomial $H_n(y)$. If $\I_i=[0,2\pi)$ then a suitable basis could be a trigonometric polynomial basis $\{e^{in\theta}\}_{n\in\Z}$. (Note that here we implicitly assume that the above polynomials are normalized.)
\bibitem{St}
From \eqref{hep} one sees that $\Phi^*(B)=Y^*(B\otimes I)Y$ (the Stinespring form) where $Y:=\sum_{n,s}\kb{e_s\otimes v_n^s}{e_n'}$ is an isometry (i.e.\ $Y^*Y=I$).
\end{thebibliography}
\end{document}